\documentclass[12pt]{article}
\usepackage{xcolor}
\usepackage{textcomp}
\usepackage{chet}
\usepackage[numbers,sort&compress]{natbib}
\usepackage{mciteplus}
\usepackage{amssymb,amsmath}

\newcommand{\cV}{\mathcal{V}}

\newcommand{\bR}{\mathbb{R}} 
\newcommand{\be}{\begin{equation}}
\newcommand{\ee}{\end{equation}}
\newcommand{\beq}{\begin{eqnarray}}
\newcommand{\eeq}{\end{eqnarray}}
\usepackage{tikz}
\usepackage{subcaption}

\usepackage{amsthm}

\newtheorem{theorem}{Theorem}[section]

\newcommand{\bZ}{\mathbb{Z}}

\usepackage[normalem]{ulem}

\title{\LARGE Small Vacuum Energy and Tunneling \\  \vspace{.3cm}  in a Modified Bousso-Polchinski Model}

\author{$^{a,b}$James Halverson, $^{c}$Justin Khoury, and $^{d}$Cody Long}

\affiliation{    $^{a}$Department of Physics, Northeastern University, Boston, MA 02115
\smallskip

$^{b}$The NSF AI Institute for Artificial Intelligence and Fundamental Interactions

\smallskip
$^{c}$Center for Particle Cosmology, Department of Physics and Astronomy, 
\\ University of Pennsylvania,  Philadelphia, PA 19104

\smallskip
$^{d}$Boulder Institute for Theoretical Physics, Boulder, CO 80301}

\abstract{We propose a simplified model for the cosmological constant in string theory flux vacua motivated by type IIB and F-theory compactifications.  Relative to the Bousso-Polchinski model, small vacuum energy spacing occurs in thin wafers rather than thin shells. The model is applied to the entire Sch\"oller-Skarke database of Calabi-Yau fourfolds, which exhibit $532,600,483$ distinct sets of Hodge numbers. The overwhelming majority of those ($99.95\%$ percent for some choices of parameters) exhibit a vacuum energy spacing of~$10^{-120}$ in Planck units or smaller.  Brown-Teitelboim membrane nucleation transitions can populate this landscape of flux vacua. In the thin-wall approximation, and ignoring gravitational corrections, we find that the bubble transitions are always dominated by giant leaps in flux space. The age of the universe places a bound on Calabi-Yau topology that is satisfied for the entire Sch\"oller-Skarke database.}

\begin{document}
\maketitle

\tableofcontents\newpage

\section{Introduction}
\label{sec:intro}

Flux compactifications of string theory yield a discretuum of vacuum energies and rich cosmological dynamics that can potentially account for the observed small value of the cosmological constant. In~\cite{Bousso:2000xa}, Bousso and Polchinski (BP) proposed that such a discretuum, together with Brown-Teitelboim membrane nucleation \cite{Brown:1988kg}, would lead to a bubble cosmology with a non-trivial distribution of cosmological constants. 
The GKP solution \cite{Giddings:2001yu} for flux vacua, and moduli stabilization scenarios that utilize it, such as KKLT \cite{Kachru:2003aw} and LVS \cite{Balasubramanian:2005zx}, provided a concrete realization of the broad BP idea in Type IIB compactifications. Estimates of the number of flux vacua arising from these solutions are $\sim 10^{500}$ \cite{Ashok:2003gk,Denef:2004cf,Denef:2004ze} and $\sim 10^{272,000}$ \cite{Taylor:2015xtz} in Type IIB and F-theory compactifications, respectively.\footnote{More recently, it has been appreciated that there is also an exponentially large number of topologically distinct string geometries, {\it e.g.}, an estimate of $10^{3000}$~\cite{Halverson:2017ffz} and lower bound of $10^{755}$~\cite{Taylor:2017yqr} in certain F-theory constructions.} Furthermore, there has been significant progress in flux compactifications with large numbers of moduli, {\it e.g.},~\cite{demirtas2021small,Demirtas:2022hqf,McAllister:2024lnt}. 

Despite this significant progress over the last $25$ years, the full problem of taming the landscape in order to make statistical predictions from string theory is still out of reach. For instance, aside from our inability to handle the size of the landscape, issues of computational complexity~\cite{Denef:2006ad,Cvetic:2010ky,Bao:2017thx,Halverson:2018cio,Halverson:2019vmd} abound and much remains to be understood about the cosmological measure~\cite{Linde:1993nz,Linde:1993xx,Garriga:2005av,Guth:2007ng} (see~\cite{Freivogel:2011eg} for a review); though there has been progress with deep learning techniques ({\it e.g.}, ~\cite{He:2017set,Ruehle:2017mzq, Krefl:2017yox,Carifio:2017bov} and ~\cite{Halverson:2019tkf,Cole:2019enn,Constantin:2021for,Abel:2021rrj,cole2021probingstructurestringtheory,Abel:2023zwg,Berglund:2024reu} with RL) and measure proposals~\cite{Denef:2017cxt,Carifio:2017nyb,Khoury:2019yoo,Khoury:2019ajl,Kartvelishvili:2020thd,Khoury:2022ish,Khoury:2023ktz,Hassfeld:2022vzk,Hassfeld:2024ake}, some of which investigate the role of complexity in the measure itself. Furthermore, a number of aspects of $\mathcal{N}=1$ corrections remain to be understood. Avoiding these eventually necessary difficulties is one of the hallmarks of the BP model: it gave a qualitative understanding of the relevant physics in a simplified model, setting the stage for the aforementioned works. 

In this paper we propose a modified BP model motivated by Type IIB and F-theory compactifications. Our model is inspired by seminal works~\cite{Denef:2004ze, Denef:2004cf} of Denef and Douglas, but is simplified in the hope that it will shed conceptual light on issues related to vacuum energy spacing, the dynamics of the landscape, and the measure problem. The model is 
\begin{equation}
V = \Lambda_0 - \frac{3}{\mathcal{V}^2}\, |N_I \Pi_I|^2,
\end{equation}
 where~$N_I$ are flux quanta,~$\Pi_I$ are periods of the holomorphic form on the Calabi-Yau,~$\mathcal{V}$ is the overall volume, and~$\Lambda_0>0$ is an uplift term. In the spirit of BP, we will ignore moduli dependence and treat these quantities as independent in an effort to understand the basic physics, but in the conclusions will discuss a number of ways to make further progress by relaxing the assumptions. Specifically, we will carry out an analysis of the vacuum energy spacing akin to the BP result, but necessitating a saddlepoint computation for the number of flux vacua rather than a volume approximation. The result is that this model can easily account for the observed small value of the cosmological constant --- in the Sch\"oller-Skarke ensemble~\cite{Scholler:2018apc} of $532,600,483$ sets of Hodge numbers for topologically rich Calabi-Yau fourfolds, a sufficiently small vacuum energy is realized an overwhelming fraction of the time, specifically $99.95\%$ for particular values of the parameters.  We then carry out an analysis of decay rates due to Brown-Teitelboim membrane nucleation in the thin-wall approximation, demonstrating that the bounce action is maximized along a particular hyperplane in flux space. Moving asymptotically away from that hyperplane yields the dominant decays, which we demonstrate arise from ``giant leaps'' in flux space~\cite{Brown:2010mg}.

This paper is organized as follows. In Sec.~\ref{sec:genBP} we review aspects of the BP model and propose a modification motivated by type IIB and F-theory considerations. In Sec.~\ref{VES} we carry out an analysis of the vacuum energy spacing, both in a volume approximation and in a saddlepoint approximation, and apply the latter to the Sch\"oller-Skarke ensemble of Calabi-Yau fourfolds. In Sec.~\ref{sec:cosmo} we turn to cosmological transitions on this flux lattice. Specifically, we study the Brown-Teitelboim membrane nucleation mechanism and decay rates in the model, and argue that the dominant decay channels generically describe giant leaps in flux space. We summarize our results and discuss some future directions in Sec.~\ref{sec:conclusions}. A number of technical appendices related to Dirichlet's approximation theorem and lattice point counting in high-dimensional spheres are essential for obtaining the results.

\section{A Modified Bousso-Polchinski Model}
\label{sec:genBP}

Our goal in this Section is to modify the BP model by including some details motivated by the GKP solution \cite{Giddings:2001yu} and type IIB/F-theory moduli stabilization scenarios, such as KKLT~\cite{Kachru:2003aw} and LVS~\cite{Balasubramanian:2005zx}.

Let us being by reviewing the BP model and some of its subtleties as a toy model for IIB moduli stabilization scenarios. It posits that the vacuum energy is 
\begin{equation}
  V_{\rm BP} = \Lambda_0 +  \frac{1}{2} \sum_{I,J=1}^{\cal J} g_{IJ}  N^I N^J\,,
\label{BP V}
\end{equation}
where $\Lambda_0 < 0$ is a bare negative cosmological constant, and $g_{IJ}$ is a metric on the vector of fluxes~$N^I$, which is of dimension~${\cal J}$. 
The original BP model~\cite{Bousso:2000xa} corresponds to~$g_{IJ} = q_I \delta_{IJ}$, with the charges~$q_I$ setting the lattice spacing,
such that
\be
V_{\rm BP} =  \Lambda_0 + \frac12 \sum_{I=1}^{\cal J} q_I^2 N_I^2 \,.
\label{BP V original}
\ee
Equation~\eqref{BP V} is a generalization proposed in~\cite{Denef:2008wq}.
Obtaining a small cosmological constant satisfying $0 < V < \epsilon$ requires finding fluxes within a shell of width~$\epsilon$ on the boundary of an~${\cal J}$-ball of radius~$|\Lambda_0|$ in the $g_{IJ}$ metric. This is depicted in the left panel of Fig.~\ref{fig:wafer}. 
The central idea of the model is that small cosmological constants exist in the shell, which actually make up the bulk of the volume of the~${\cal J}$-ball for large~${\cal J}$. Transitions between these fluxes are mediated by the Brown-Teitelboim membrane nucleation mechanism, generalizing the Brown-Teitelboim result to the case of many fluxes.
The model ignores both the dependence of~$g_{IJ}$ on moduli fields, and also consistency conditions such as Gauss's laws that depend on the context.

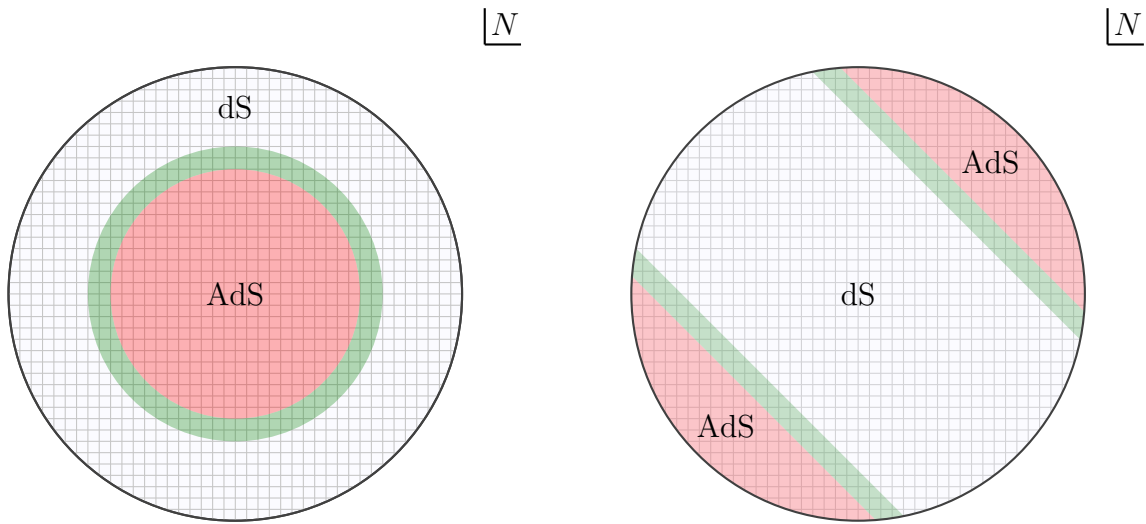
\begin{figure}[h]
\centering
\begin{tikzpicture}[scale=1.5]
    \draw[thick, darkgray, fill=blue!5, fill opacity=0.3] (0,0) circle (2cm);

    \begin{scope}
        \clip (0,0) circle (2cm);
        \foreach \x in {-2.5,-2.4,...,2.5} {
            \draw[gray!45, thin] (\x,-2.5) -- (\x,2.5);
        }
        \foreach \y in {-2.5,-2.4,...,2.5} {
            \draw[gray!45, thin] (-2.5,\y) -- (2.5,\y);
        }
    \end{scope}

    \fill[green!50!black, opacity=0.3, even odd rule] (0,0) circle (1.3cm) (0,0) circle (1.1cm);

    \fill[red, opacity=0.3] (0,0) circle (1.1cm);

    \draw[thick, darkgray] (0,0) circle (2cm);

    \draw[thick] (2.2,2.2) -- (2.52,2.2);
    \draw[thick] (2.2,2.2) -- (2.2,2.52);
    \node at (2.38,2.38) {$N$};

    \node[] at (0, 0) {AdS};
    \node[] at (0, 1.65) {dS};
\end{tikzpicture}
\hspace{1cm}
\begin{tikzpicture}[scale=1.5]
    \def\rotangle{45}

    \begin{scope}
        \clip (0,0) circle (2cm);

        \foreach \x in {-2.5,-2.4,...,2.5} {
            \draw[gray!45, thin] (\x,-2.5) -- (\x,2.5);
        }
        \foreach \y in {-2.5,-2.4,...,2.5} {
            \draw[gray!45, thin] (-2.5,\y) -- (2.5,\y);
        }

        \begin{scope}[rotate=\rotangle]
            \fill[red, opacity=0.3] (1.3,-2.5) rectangle (2.5,2.5);
            \fill[red, opacity=0.3] (-2.5,-2.5) rectangle (-1.3,2.5);
        \end{scope}

        \begin{scope}[rotate=\rotangle]
            \fill[green!50!black, opacity=0.3] (1.1,-2.5) rectangle (1.3,2.5);
        \end{scope}

        \begin{scope}[rotate=\rotangle]
            \fill[green!50!black, opacity=0.3] (-1.3,-2.5) rectangle (-1.1,2.5);
        \end{scope}
    \end{scope}

    \draw[thick, darkgray, fill=blue!5, fill opacity=0.3] (0,0) circle (2cm);

    \draw[thick] (2.2,2.2) -- (2.52,2.2);
    \draw[thick] (2.2,2.2) -- (2.2,2.52);
    \node at (2.38,2.38) {$N$};

    \begin{scope}[rotate=\rotangle]
        \node[] at (1.65, 0) {AdS};
        \node[] at (-1.65, 0) {AdS};
        \node[] at (0, 0) {dS};
    \end{scope}
\end{tikzpicture}
\caption{{\it Left:} The Bousso-Polchinski model in flux space. The interior red region contains AdS vacua, surrounded by a thin green shell of de Sitter vacua with cosmological constant~$0 \leq V \leq \Delta\Lambda$, and an outer dS region. {\it Right:} Our model in flux space. The small-$\Lambda$ dS regions are organized into wafer-shaped slices, with dS instead of AdS in the central region.}
\label{fig:wafer}
\end{figure}

We first revisit the model in light of the GKP solution and its utilization in type IIB moduli stabilization scenarios such as KKLT and LVS, especially since the BP model predates the others. This was done in detail in seminal works by Ashok, Denef, and Douglas \cite{Ashok:2003gk,Denef:2004cf,Denef:2004ze}, but here we will aim to retain a simple toy model in the spirit of the original BP proposal, with an eye toward subsequent tunneling and measure studies. 

In the F-theory description of GKP on an elliptic Calabi-Yau fourfold $X$, the vacuum energy arises as 
\begin{equation}
V = -\frac{\chi}{24} + N_{D3} + \frac12 \int_X G_4 \wedge *G_4\,,
\label{pot start}
\end{equation}
where~$\chi$ is the Euler characteristic of~$X$, and~$N_{D3}$ is the number of D3-branes. 
Expanding~$G_4$ in a basis of harmonic forms~$\Sigma_I \in H^4(X,\bZ)$, this realizes the BP model provided that 
\begin{equation}
\Lambda_0 = -\frac{\chi}{24} + N_{D3}\,; \qquad \qquad g_{IJ} = \int \Sigma_I \wedge * \Sigma_J\,.
\label{Lambda0 chi}
\end{equation}
The metric dependence of~$g_{IJ}$, neglected in the BP model, enters through the Hodge star~$*$. Crucially, the D3-brane tadpole cancellation condition,
\begin{equation}
-\frac{\chi}{24} + N_{D3} + \frac12 \int G_4 \wedge G_4 = 0\,,
\label{eqn:d3tadpole}
\end{equation}
was also neglected. This condition may be used to rewrite the potential~\eqref{pot start} as 
\begin{equation}
V = \frac{1}{4} \int \big(G_4 - * G_4\big) \wedge *\big(G_4 - * G_4\big)\,.
\end{equation}
At this intermediate stage, there are a number of potential causes for concern. One might worry that since the intersection pairing on~$H^4(X,\bZ)$ is not positive definite in general, there may be an infinite number of solutions to the D3-brane tadpole cancellation condition. However, for self-dual vacuum solutions,~$G_4 = * G_4$, we have~$\int G_4\wedge G_4 = \int G_4 \wedge *G_4$, which is positive definite. We also have a marked departure from the original BP model: the D3-brane tadpole cancellation condition \eqref{eqn:d3tadpole} renders~$\Lambda_0$ flux-dependent, unlike the BP model, and furthermore in such a way that changing from one self-dual flux to another does not change the potential! This is the standard SUGRA Minkowski solution after complex structure moduli stabilization with $W_0\neq 0$. 

As GKP demonstrated, this potential arises from the contribution to the~$\mathcal{N}=1$ F-term potential from the Gukov-Vafa-Witten superpotential
\begin{equation}
W = \int G_4 \wedge \Omega = N_I \Pi_I\,,
\end{equation}
where~$\Omega$ is the holomorphic 4-form, and
\begin{equation}
\Pi_I = \ell_s^{-3} \int_{\Sigma_I} \Omega
\end{equation}
are the periods in units of the string length~$\ell_s$. In a slight abuse of notation, we write~$\Sigma_I$ also as a four-cycle by Poincar\' e duality.
We follow the conventions of~\cite{demirtas2021small} of dimensionless superpotential for the flux superpotential, since the periods are in flat coordinates. 
We define 
\be
{\cal V} = \frac{\text{Vol}(X)}{\ell_s^{6}}  
\label{vol reln}
\ee
the volume of the Calabi-Yau in string units. 

Let us review how moduli stabilization scenarios such as KKLT alter the vacuum structure via the inclusion of additional effects beyond complex structure stabilization. First, K\" ahler moduli stablization puts the theory in a SUSY AdS minimum whose depth depends on~$W_0$, 
\begin{equation}
  V_{\rm AdS} = -3 {\rm e}^{K/M_{\rm Pl}} |W_0|^2M_{\rm Pl}^4\,.
\end{equation}
Setting~$M_{\rm Pl} = 1$, and using the tree-level K\"ahler potential~$K=-2\log\cV$ with string loop and $\alpha'$ corrections neglected, we have
\begin{equation}
V_{\rm AdS} = -\frac{3}{\cV^2} \,|N_I \Pi_I|^2.
\end{equation}
In a full treatment, the periods are evaluated at vacua in complex structure moduli space. If one wishes, this can be written in terms of a ``metric" as 
\begin{equation}
  V_{\rm AdS} = -\frac{3}{\cV^2} \, g_{IJ} N^I N^J\,; \qquad g_{IJ} = \Pi_I \overline{\Pi}_J\,.
\end{equation}Notably, increasing the length of the flux vector in this metric at fixed complex structure \emph{lowers} the vacuum energy (unless the moduli are tuned to allow for a non-trivial zero eigenvector of~$g_{IJ}$). 

The final part of moduli stabilization schemes is a mechanism for uplift to de Sitter space, {\it e.g.}, via anti-D3 branes \cite{Kachru:2003aw} or D-terms \cite{burgess2003sitter}. Uplifts can be difficult to control, but in a our toy model we simply recognize the phenomenological necessity of positive vacuum energy and accordingly add a \emph{positive} $\Lambda_0$. A key assumption of our toy model with respect to the uplift is in treating it as independent of complex structure moduli \cite{Kachru_2007,Balasubramanian:2005zx}. We keep in mind that $\Lambda_0$ should vanish in the decompactification limit, so one should consider 
\begin{equation}
    \Lambda_0 \mapsto \frac{\tilde \Lambda_0}{\mathcal{V}^{2-\alpha}} \,,
\end{equation}
where $\alpha=0$ for an unwarped anti-D3 branes or D-terms and $\alpha=2/3$ for warped anti-D3 branes \cite{Kachru:2003sx}. For simplicity we leave $\Lambda_0$ in the following, remembering the $\mathcal{V}$ dependence, and take $\tilde \Lambda_0$ to be a constant scale.

\bigskip
Summarizing our construction, considerations from moduli stabilization schemes in type IIB/F-theory compactifications suggest a modification of the BP model to 
\begin{equation}
\boxed{V = \Lambda_0 -  \frac{3}{\cV^2} \,|N_I \Pi_I|^2\,.}
\label{our V}
\end{equation}
Compared to~$V_{\rm BP}$ in Eq.~\eqref{BP V}, there are two key differences: i) there is now a {\it minus sign} on the flux contribution;~ii) we have a positive~$\Lambda_0$ instead of negative~$\Lambda_0$, to do the uplift. 

Thus, instead of increasing flux raising the cosmological constant from a negative value as in the original BP model, we are lowering it from a positive value. Again considering the desire to choose fluxes such that we are in a low-lying de Sitter vacuum,~$0 < V < \epsilon$, the picture is now as in the right panel of Fig.~\ref{fig:wafer}. The central aspects of the BP model---a discretuum that can give rise to small cosmological constants and transitions between them mediated by the Brown-Teitelboim mechanism---are preserved, but now increasing flux numbers means decreasing vacuum energy, rather than larger vacuum energy, which will affect bubble nucleation and cosmological measure considerations.

\section{Vacuum Energy Spacing}
\label{VES}
 
Following BP, the question we would like to address is whether the discretuum of flux vacua, with potential energy given by~\eqref{our V}, is dense enough to achieve
a vacuum energy with the requisite sensitivity. Specifically, we are interested in the existence of flux vacua with vacuum energy sufficiently close to a fixed value~$V_\star$. For concreteness, we choose the cosmologically relevant value of~$V_\star = 0$, such that
\begin{equation}
0 \leq V_N \leq  \Delta\Lambda\,,
\label{VN bound}
\end{equation} 
where~$V_N$ is a model for cosmological constants that depends on flux quanta~$N$, and~$\Delta\Lambda \sim 10^{-120}$  
roughly matches the observed value in Planck units. 

Let us briefly review how this is achieved in the BP mechanism, focusing for simplicity on the original potential energy~\eqref{BP V original}.
This equation, together with~\eqref{VN bound}, yields
\be
2|\Lambda_0| <   \sum_{I=1}^{\cal J} q_I^2 N_I^2 < 2 \Big(|\Lambda_0| + \Delta\Lambda\Big)\,.
\label{shell BP}
\ee
To linear order in~$\frac{\Delta \Lambda}{|\Lambda_0|} \ll 1$, the shell has radius~$R = \sqrt{2|\Lambda_0|}$ and width
\be
\Delta R = \frac{\Delta\Lambda}{\sqrt{2|\Lambda_0|}}\,.
\ee
Thus its volume is given by 
\be
{\cal V}_{\rm shell} = {\cal V}_{}\frac{2\pi^{{\cal J}/2}}{\Gamma({\cal J}/2)} \big(2|\Lambda_0|\big)^{\frac{{\cal J}}{2} - 1}\Delta\Lambda\,,
\ee 
where~$2\pi^{{\cal J}/2}/\Gamma({\cal J}/2)$ is the area of the~${\cal J}-1$ sphere. The number of lattice points within the shell is the product of this volume times the density of lattice points,~$1/\prod\limits_{I = 1}^{{\cal J}} q_I$. Therefore, the requirement that there exist at least one lattice point within the shell,~${\cal V}_{\rm shell}
> \prod\limits_{I = 1}^{{\cal J}}q_I$, translates to the lower bound:
\be
\frac{\Delta \Lambda}{|\Lambda_0|} > \frac{\Gamma({\cal J}/2)}{\pi^{{\cal J}/2}} \prod\limits_{I = 1}^{{\cal J}}  \frac{q_I}{\sqrt{2|\Lambda_0|}} \,.
\ee
Using Stirling's formula for~$\Gamma({\cal J}/2)$ in the large-${\cal J}$ limit, the inequality simplifies to
\be
\boxed{\frac{\Delta \Lambda}{|\Lambda_0|}\, \gtrsim\, 2{\rm e} \sqrt{\frac{\pi}{{\cal J}}} \prod\limits_{I = 1}^{{\cal J}} \sqrt{\frac{q_I^2 {\cal J}}{4\pi {\rm e} |\Lambda_0|}} \qquad \text{(BP spacing)}\,.}
\label{BP Del Lambda}
\ee
Assuming~$|\Lambda_0|\sim 1$ in Planck units, such that~$\frac{\Delta \Lambda}{2|\Lambda_0|} \sim 10^{-120}$, then the inequality
can be satisfied for~${\cal J} \sim {\cal O}(10^2)$ and~$q_I \lesssim {\cal J}^{-1/2} \sim 10^{-1}$. For instance, with~${\cal J} = 200$ and~$q = 0.1$, the right-hand side is~$\simeq 2\times 10^{-124}$. 

The result depends on the assumption that the number of fluxes associated to lattice points in a high-dimensional ball is well approximated by the volume of the ball.  This depends crucially on both the radius and the dimension, and can be violated in practice \cite{Denef:2008wq}, {\it e.g.}, in the F-theory geometry with the largest number of flux vacua \cite{Taylor:2015xtz}. The volume approximation and improvements upon it will play a crucial role in our results.

\bigskip
We want to see if a similar argument applies to our case. Given the potential energy~\eqref{our V}, the analog of~\eqref{shell BP} is
\be
\Lambda_0 - \Delta \Lambda < \frac{3}{{\cal V}^2}\left\vert \Pi_I N^I\right\vert^2 <  \Lambda_0\,,
\label{shell ours}
\ee
where we recall that~$\Lambda_0 > 0$ in our case. By D3 tadpole cancellation~\eqref{eqn:d3tadpole} and $N_\text{D3}\geq0$, the flux vector has a norm bounded by the Euler characteristic:
\be
\frac12 N_I N^I \leq \frac{\chi}{24}\,,
\label{N bound}
\ee
where we have used flux-space coordinates where~$g_{IJ}=\delta_{IJ}$ for self-dual fluxes. 
This choice introduces a Jacobian factor in the density of lattice points, see, {\it e.g.}, \cite{Denef:2008wq} for a treatment. For self-dual fluxes it arises from the positive-definite part of the intersection pairing on the middle homology of the four-fold, which we neglect since it is highly model-dependent.

To simplify the discussion, let us follow BP and ignore moduli dependence, in which case we take~$\Pi_I$ to be a constant vector independent of~$N_I$. Since an overall phase may be absorbed into the definition of~$\Omega$, we may also choose the overall phase of~$\Pi_I$ such that the real and imaginary components of~$\Pi_I$ are equal,
that is,
\be
{\rm Re}\, \Pi_I = {\rm Im}\, \Pi_I \equiv \frac{1}{\sqrt{2}} \hat{\Pi}_I\,.
\label{Re Im Pi}
\ee
By definition,~$\hat{\Pi}_I$ is a real-valued vector, and we shall denote its norm by~$\hat{\Pi} = \sqrt{\hat{\Pi}_I\hat{\Pi}^I}$. Thus~\eqref{shell ours} becomes
\be
\Lambda_0 - \Delta \Lambda < \frac{3}{{\cal V}^2} \big(\hat{\Pi} \cdot N\big)^2 < \Lambda_0\,.  
\label{slab}
\ee
At this step in the calculation, the BP model would identify equipotentials --- the geometry in flux space with fixed vacuum energy, which in their case is a~${\cal J}-1$ sphere. We do the same, and find a rather different result: equipotentials with vacuum energy~$V_\star$ are given by the hyperplanes 
\begin{equation}
\hat \Pi\cdot N = \pm \sqrt{\frac{{\cal V}^2 \big(\Lambda_0-V_\star\big)}{3}  } \,.
\label{hyperplane eqn}
\end{equation} 
Such hyperplanes have a unit normal~$n_I = \frac{\hat{\Pi}_I}{\hat{\Pi}}$,
and are translated from the origin along the normal by a distance~$\sqrt{\frac{{\cal V}^2 (\Lambda_0-V_\star)}{3\hat \Pi^2}}$ in either direction.
The equipotentials of interest for our calculation have~$V_\star = 0$. See Fig.~\ref{fig:wafer_geometry}.

\begin{figure}[h]
\centering
\begin{tikzpicture}[scale=1.5]
    \def\rotangle{45}

    \begin{scope}
        \clip (0,0) circle (2cm);

        \foreach \x in {-2.5,-2.4,...,2.5} {
            \draw[gray!45, thin] (\x,-2.5) -- (\x,2.5);
        }
        \foreach \y in {-2.5,-2.4,...,2.5} {
            \draw[gray!45, thin] (-2.5,\y) -- (2.5,\y);
        }

        \begin{scope}[rotate=\rotangle]
            \fill[red, opacity=0.3] (1.3,-2.5) rectangle (2.5,2.5);
            \fill[red, opacity=0.3] (-2.5,-2.5) rectangle (-1.3,2.5);
        \end{scope}

        \begin{scope}[rotate=\rotangle]
            \fill[green!50!black, opacity=0.3] (1.1,-2.5) rectangle (1.3,2.5);
        \end{scope}

        \begin{scope}[rotate=\rotangle]
            \fill[green!50!black, opacity=0.3] (-1.3,-2.5) rectangle (-1.1,2.5);
        \end{scope}
    \end{scope}

    \draw[thick, darkgray, fill=blue!5, fill opacity=0.3] (0,0) circle (2cm);

    \draw[thick] (2.2,2.2) -- (2.52,2.2);
    \draw[thick] (2.2,2.2) -- (2.2,2.52);
    \node at (2.38,2.38) {$N$};

    \begin{scope}[rotate=\rotangle]
        \draw[<->, thick] (1.11, -1.15) -- (1.29, -1.15);
        \node[font=\small] at (1.2, -1.35) {$w$};
    \end{scope}

    \begin{scope}[rotate=\rotangle]
        \draw[-, thick, gray] (0,0) -- (1.2, 1.6);
        \node[font=\small] at (0.3, 1.0) {$\sqrt{\frac{\chi}{12}}$};

        \draw[-, thick, gray] (0,0) -- (1.2, 0);
        \node[font=\small] at (0.6, -0.35) {$\sqrt{\frac{\cV^2 \Lambda_0}{3 \,\hat{\Pi}^2}}$};

        \draw[-, thick, gray] (1.2, 0) -- (1.2, 1.6);
        \node[font=\small] at (1.5, 0.6) {$R_\text{disk}$};
    \end{scope}
\end{tikzpicture}
\caption{Geometric quantities determining flux counts. The tadpole constraint sets the ball radius~$\sqrt{\chi/12}$, while the distance from the origin to the thickness-$w$ wafer is~$\sqrt{\cV^2 \Lambda_0 / (3\,\hat{\Pi}^2)}$. $R_\text{disk}$ controls the number of lattice points and flux vacua with small cosmological constant.}
\label{fig:wafer_geometry}
\end{figure}
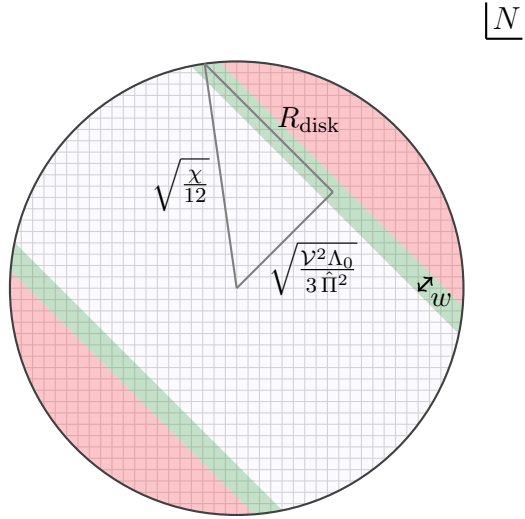

Together, the hyperplane equation~\eqref{hyperplane eqn} and tadpole condition~\eqref{N bound}
restrict the allowed set of lattice points to an~$({\cal J}-1)$-dimensional disk, of radius 
\be
R_{\rm disk} = \sqrt{ \frac{\chi}{12} - \frac{{\cal V}^2 \Lambda_0}{3 \,\hat{\Pi}^2}}\,.
\ee 
More generally, Eq.~\eqref{slab}, together with~\eqref{N bound}, defines a thin wafer of width 
\be
w = \frac{\Delta \Lambda}{ 2\hat{\Pi} \sqrt{\Lambda_0} } \sqrt{\frac{{\cal V}^2}{3}}\,.
\ee
The volume of the wafer is given by the width times the area of the~${\cal J}-1$-dimensional disk:
\be
{\cal V}_{\rm wafer} = \frac{2\pi^{\frac{{\cal J}-1}{2}}}{({\cal J}-1)\Gamma\left(\frac{{\cal J}-1}{2}\right)} \left( \frac{\chi}{12} - \frac{{\cal V}^2 \Lambda_0}{3\, \hat{\Pi}^2} \right)^{\frac{{\cal J}-1}{2}}\frac{\Delta \Lambda}{ 2\hat{\Pi} \sqrt{\Lambda_0} } \sqrt{\frac{{\cal V}^2}{3}}\,.
\ee
Using Stirling's formula and the volume approximation (which will be relaxed later), the requirement that there exist at least one lattice point within the wafer,~${\cal V}_{\rm wafer} > 1$, translates to
\be
\frac{\Delta \Lambda}{\Lambda_0} \gtrsim \sqrt{2{\rm e}} \pi \left(\frac{6{\cal J}}{{\rm e}\pi \chi}\right)^{{\cal J}/2}  \frac{\sqrt{\frac{\hat{\Pi}^2 \chi}{{\cal V}^2 \Lambda_0}}}{\left(1 - \frac{4{\cal V}^2\Lambda_0}{\hat{\Pi}^2\chi}\right)^{\frac{{\cal J}-1}{2}}} > 2\sqrt{2{\rm e}} \pi \left(\frac{6{\cal J}}{{\rm e}\pi \chi}\right)^{{\cal J}/2} \sqrt{\frac{\hat{\Pi}^2 \chi}{4{\cal V}^2 \Lambda_0}}\,.
\label{our Del Lambda bound 1}
\ee
The last factor must satisfy~$\frac{\hat{\Pi}^2 \chi}{4{\cal V}^2 \Lambda_0} > 1$, in order for~$R_{\rm disk}$ to be real.
Thus our bound becomes 
\be
\boxed{\frac{\Delta \Lambda}{\Lambda_0} \gtrsim 2\sqrt{2{\rm e}} \pi\left(\frac{6{\cal J}}{{\rm e}\pi \chi}\right)^{{\cal J}/2} \,.}
\label{eqn:dellam_over_lam0_volume}
\ee
Interestingly, compared to the BP result~\eqref{BP Del Lambda}, we see that~${\cal J}/\chi$ effectively plays the role of~$q_I^2 {\cal J}$ in our case. A sufficiently small vacuum energy spacing,~$\Delta \Lambda\lesssim 10^{-120}$ requires
\be
 \Lambda_0  \left(\frac{6{\cal J}}{{\rm e}\pi \chi}\right)^{{\cal J}/2}\lesssim 6 \times 10^{-122} \,.
\label{eqn:dellam_over_lam0_volume final}
\ee
This result can be achieved readily, e.g. for~${\cal J} \sim 200$ with~${\cal J}/\chi \lesssim .1$ and $\Lambda_0 < 1$, but is conditional on the volume approximation.

The above derivation relied on the volume approximation (as in BP) to the number of lattice points (or vacua) within it. Its virtue is that it gives a simple analytic expression, which makes the~$\mathcal{J}$-dependence manifest. Below we will go beyond the volume approximation using a saddlepoint method. In Appendix~\ref{app:volume_approximation} we will see that the volume approximation is reasonable for~${\cal J}/\chi < 1/3$. Relative to Eq.~\eqref{eqn:dellam_over_lam0_volume final}, this implies that~$\Delta \Lambda / \Lambda_0$ is very small when the volume approximation is valid. 

\subsection{Improved Vacuum Counts with a Saddlepoint Approximation}

We wish to improve the analysis by generalizing it to the case that the volume approximation no longer holds. 
Specifically, whether the volume of a~$d$-ball of radius~$R$ is a good approximation to the number of lattice points within it depends on~$d$ and~$R$. 
 We will see examples in which it fails, and therefore we wish to consider more accurate methods~\cite{Denef:2008wq}. 

Let~$N(d,R)$ be the number of lattice points in a~$d$-ball~$B_{d,R}$ of radius~$R$. This is equal to \cite{Mazo1990}
\begin{equation}
N(d,R) = \frac{1}{2\pi {\rm i}} \int \frac{{\rm d}t}{t} {\rm e}^{-tR^2/2}\, \, \vartheta_3\big(0,{\rm e}^{t/2}\big)^d\,,
\end{equation}
where~$\vartheta_3$ is the Jacobi theta function. When~$d$ is large, the integral may be evaluated in the saddlepoint approximation,
\begin{equation}
N_{\rm saddle} (d,R) := N(d,R) \big\vert_\text{saddle} = {\rm e}^{S{(t_\star)}} \,,
\end{equation}
where~$t_\star$ is the saddle point of the integrand, given by the extrema of
\begin{equation}
S(t) = -\ln(-t) + d \,\ln \vartheta_3\big(0,{\rm e}^{t/2}\big) - \frac{tR^2}{2}\,.
\label{eqn:saddle_point_action}
\end{equation}
In the context of counting string vacua this is discussed at length in~\cite{Denef:2008wq}; see~\cite{Mazo1990} for original derivations in the math literature. 

Some well-known results in F-theory rely heavily on this method. Let the volume approximation be~$N_{\rm vol}(d,R) := \text{vol}(B_{d,R})$ for the sake of comparison. 
For concrete examples, we consider two notable F-theory geometries:
\begin{itemize}

\item The F-theory geometry~$\mathcal{M}_\text{max}$ with the most flux vacua has fourth Betti number and Euler characteristic given by
\begin{align}
\nonumber
(h^{1,1}, h^{2,1}, h^{3,1})^\text{max} &= (252,0,303148)\\ 
\nonumber
b_4^{\text{max}} &= 1819942\\
\chi^{\text{max}} &= 1820448 \,.
\end{align}

\item Its mirror, the F-theory geometry~$\mathcal{M}_{h^{1,1}\text{-max}}$ with maximal~$h^{1,1}$, has
\begin{align}
\nonumber
(h^{1,1}, h^{2,1}, h^{3,1})^{h^{1,1}\text{-max}} &= (303148,0,252)\\ 
\nonumber
b_4^{h^{1,1}\text{-max}} &= 1214150\\
\chi^{h^{1,1}\text{-max}} &= 1820448 \,.
\end{align}

\end{itemize}
The dimension of the associated ball is~$b_4$, and the radius for both is~$R = \sqrt{2Q} = \sqrt{2 \chi/24}  \approx 389.6$. This data yields estimates for the number of flux vacua in the two cases:
\begin{align}
&\mathcal{M}_\text{max}:   & N_{\rm saddle}^{\text{max}}\approx 10^{272000}\,;\qquad~~~~~   & N_{\rm vol}^{\text{max}}           \approx 10^{140000} \label{eqn:max_comparison}\\
& \mathcal{M}_{h^{1,1}\text{-max}}:   & N_{\rm saddle}^{h^{1,1}\text{-max}}\approx 10^{244000}\,; \qquad~~~~~  & N_{\rm vol}^{h^{1,1}\text{-max}}   \approx 10^{200000} \label{eqn:h11max_comparison}\,.
\end{align}
These saddle results are from~\cite{Taylor:2015xtz} and~\cite{Wang:2020gmi}, but this calculation demonstrates the importance of checking the validity of the volume approximation. The results are quite different, and in the volume counting one would conclude that the geometry with the maximal~$h^{1,1}$ is actually the F-theory geometry with the most flux vacua.

\medskip
With these caveats about the volume approximation in mind, we redo the calculation for the~$\Delta \Lambda/\Lambda_0$ bound using the saddlepoint approximation. The condition is that 
\begin{equation}
{\cal V}_{\rm wafer} =  N_{\rm saddle}(\mathcal{J}-1, R_\text{disk}) \, w> 1\,,
\end{equation} 
which yields
\begin{equation}
  \frac{\Delta \Lambda}{\Lambda_0} > \frac{2}{N_{\rm saddle} (\mathcal{J}-1, R_\text{disk})}\, \sqrt{\frac{3\hat \Pi^2}{\mathcal{V}^2 \Lambda_0}}=: \delta\,.
  \label{eqn:dellam_over_lam0_saddle}
\end{equation}
Compared to \eqref{eqn:dellam_over_lam0_volume} this equation is more accurate for large~$\mathcal{J}$, but the~$\mathcal{J}$-dependence is not as clear, due to the non-trivial dependence of the saddlepoint calculation on~$\mathcal{J}$. As we will see in examples,~$\log\delta \simeq -\log N_{\rm saddle}$ is very small.

\subsection{Vacuum Energy Spacing in the Sch\"oller-Skarke Database}

In~\cite{Scholler:2018apc}, Sch\"oller and Skarke computed all weight systems for Calabi-Yau fourfolds from reflexive polyhedra, giving rise to an enormous database of smooth Calabi-Yau fourfolds. Though significantly larger numbers of Calabi-Yau fourfolds are known to exist, {\it e.g.},~\cite{Halverson:2017ffz} and~\cite{Taylor:2017yqr}, a reflexive polytope description in~\cite{Scholler:2018apc} makes it straightforward to compute the Hodge numbers in terms of polytope data, which the authors did and have kindly shared with us. Their database realizes fourfolds with~$532,600,483$ distinct sets of Hodge numbers.

Such geometries enjoy simple relations amongst their Hodge and Betti numbers. Indeed, given~$(h^{1,1}, h^{2,1}, h^{3,1})$, the remaining Hodge numbers are given by
\begin{align}
\nonumber
h^{2,2} &= 4\big(h^{1,1} + h^{3,1}\big) + 44 - 2h^{2,1}\,;\\
\nonumber
b_4 &= 2 + 2h^{3,1} + h^{2,2} \,;\\ 
\chi &= 6\big(8 + h^{1,1} + h^{3,1} - h^{2,1}\big)\,.
\label{b4chi}
\end{align}
For the entire database we have computed the lower-bound on the vacuum energy spacing in \eqref{eqn:dellam_over_lam0_saddle}, assuming ${\cal V}^2\Lambda_0/\hat{\Pi}^2=1$ and noting that~$R = \sqrt{2\chi/24}$, 
~$\mathcal{J}=b_4$. The results are presented in Fig.~\ref{fig:scholler_skarke} and can give rise to small vacuum energy spacing.

Under Calabi-Yau fourfold mirror symmetry,~$h^{3,1}$ and~$h^{1,1}$ are interchanged while~$h^{2,1}$ is fixed. This leaves~$\chi$ and~$h^{2,2}$ invariant, but not~$b_4$. Since our vacuum energy spacing depends on~$b_4$ and~$\chi$, it is therefore not respected by mirror symmetry. This seems to contradict the RHS of Fig.~\ref{fig:scholler_skarke}. There, the classic mirror symmetry ``shield" is displayed. A careful inspection of the colors at far top and far right of the plot show points with slightly different shades of purple, corresponding to the vacuum energy spacing of~$\sim 10^{-244000}$ and~$10^{-272000}$ for the F-theory geometries with maximal~$h^{1,1}$ and most flux vacua, respectively. These results are effectively~$1/N_{\rm saddle}$ for the associated geometries, which might be understood in terms of the $\mathcal{N}=1$ breaking of mirror symmetry.

\begin{figure}
   \centering
   \includegraphics[width=1.0\textwidth]{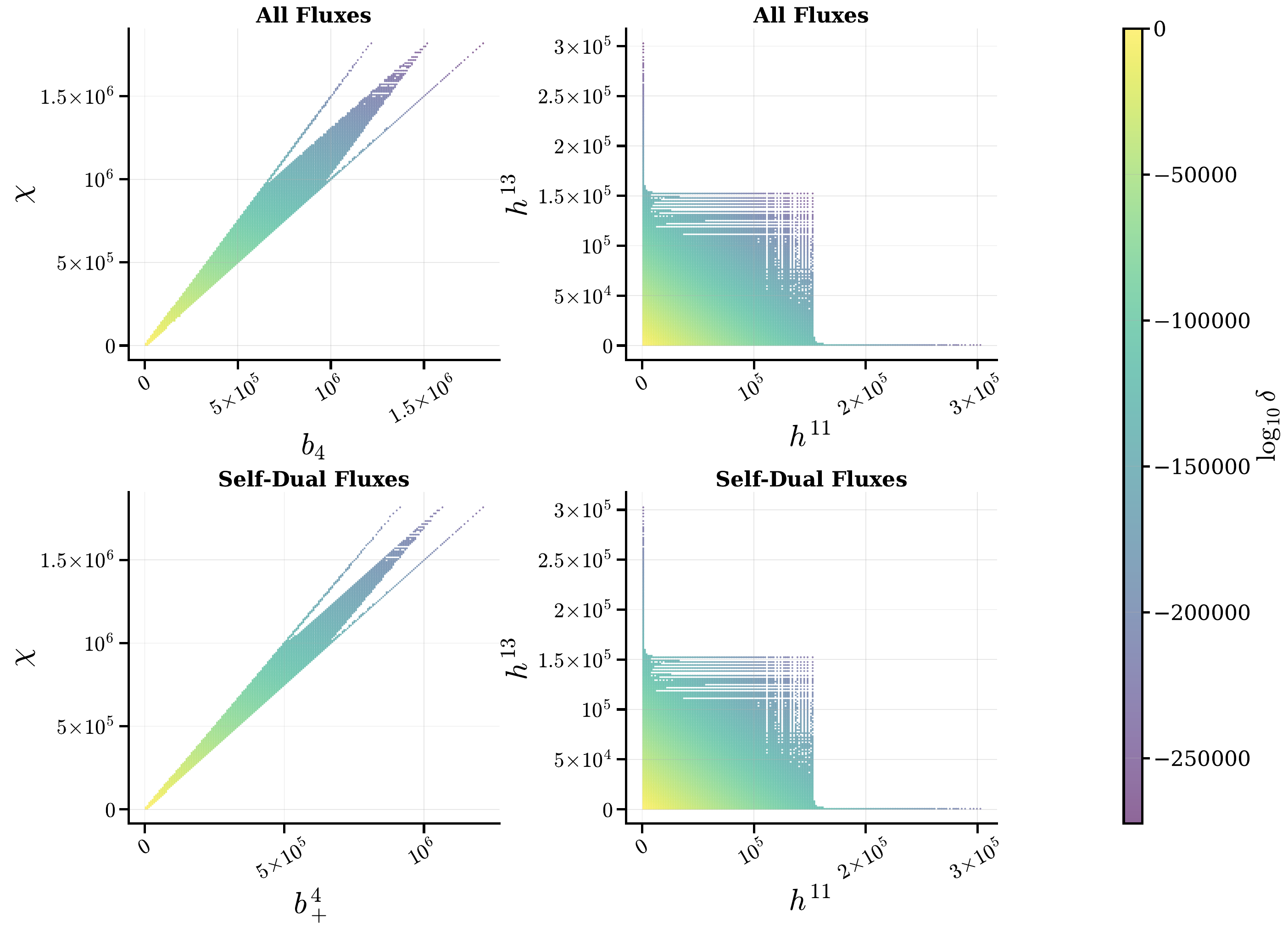}
   \caption{Vacuum energy spacing, defined as~$\frac{\Delta \Lambda}{\Lambda_0} > \delta$ in Eq.~\eqref{eqn:dellam_over_lam0_saddle}, applied to the Sch\"oller-Skarke database. {\it Top row:} For the most general four-form fluxes afforded by the theory, {\it i.e.}, with~$\mathcal{J} = b_4$, we show the spacing bound~$\delta$ as a function of~$\mathcal{J}$ and~$\chi$ ({\it Left}) and as a function of~$h^{1,1}$ and~$h^{3,1}$ ({\it Right}), which are exchanged under mirror symmetry. {\it Bottom row:} Same plots now restricted to self-dual fluxes, {\it i.e.}, with~$\mathcal{J} = b_4^+$. With all (self-dual) fluxes we find that~$99.96\%$ ($99.95\%$) of the Sch\"oller-Skarke database have~$\delta \lesssim 10^{-120}$.  } 
     \label{fig:scholler_skarke}
\end{figure}

\bigskip
\noindent \emph{Self-Dual Fluxes}

In the previous analysis of the Sch\"oller-Skarke database, we have taken~$\mathcal{J} = b_4$, {\it i.e.}, we consider the most general four-form fluxes afforded by the theory. On the other hand, fluxes used in common moduli stabilization scenarios are self-dual,~$G_4 = * G_4$, when formulated in F-theory language. We wish to perform an analysis of the Sch\" oller-Skarke database in this case as well.  

To do so, we compute the number of self-dual fluxes. The Lefschetz decomposition of~$H^4$ is given by (see, {\it e.g.}, \cite{Denef:2008wq})
\begin{align}
\nonumber
H^4_+ &= H^{0,4}_{\ell=0} \oplus H^{2,2}_{\ell=0} \oplus H^{2,2}_{\ell=2} \oplus H^{4,0}_{\ell=0}\,; \\ 
H^4_- &= H^{1,3}_{\ell=0} \oplus H^{2,2}_{\ell=1} \oplus H^{3,1}_{\ell=0}\,,
\end{align}
where~$\ell$ is the ``spin'' associated to the form, and the raising operator of the associated~$SU(2)$ algebra is given by wedging with the K\"ahler form. The self-dual fluxes are therefore given by~$H^4_+$. There is one~$\ell=2$~$(2,2)$-form, and the number of~$\ell=1$~$(2,2)$-forms is~$h^{1,1}-1$ since they are obtained by acting with the raising operator twice on any~$(1,1)$-form orthogonal to the K\" ahler form. The remaining~$(2,2)$-forms are~$\ell=0$, and there are~$h^{2,2} - h^{1,1}$ of them. The number of cohomology classes with self-dual fluxes is therefore
\begin{equation}
b_4^+ = 2 h^{4,0} + h^{2,2}_{\ell=0} + h^{2,2}_{\ell=2} = 3 + h^{2,2}-h^{1,1}\,,
\end{equation}
which can easily be determined from the Hodge numbers of the Sch\"oller-Skarke database.

We repeat the previous analysis, except in this case we take~$\mathcal{J} = b_4^+$.  The results are presented in the bottom row of Fig.~\ref{fig:scholler_skarke}. The vacuum energy spacing is larger because the dimension of the flux space is smaller. The examples~$\mathcal{M}_\text{max}$ and its mirror~$\mathcal{M}_{h^{1,1}\text{-max}}$ with the largest~$h^{1,1}$ have saddle-point estimates for the number of flux vacua given by
\begin{align}
  N_{\rm saddle}^{\text{max}}           \approx 10^{244000}\,;  \qquad \qquad
  N_{\rm saddle}^{h^{1,1}\text{-max}}   \approx 10^{224000}\,,   \label{eqn:selfdual_big_vac_counts}
\end{align}
which drives the small vacuum energy spacing. In the case where self-duality is not imposed, these should be compared respectively to~$N_{\rm saddle}^{\text{max}}\approx 10^{272000}$ and~$N_{\rm saddle}^{h^{1,1}\text{-max}}\approx 10^{244000}$ in Eqs.~\eqref{eqn:max_comparison} and~\eqref{eqn:h11max_comparison}. We have followed the tradition~\cite{Taylor:2015xtz} of keeping three significant figures in these large exponents, which has led to the coincidence that~$244,000$ arises in both~\eqref{eqn:max_comparison} and~\eqref{eqn:selfdual_big_vac_counts}, despite different geometries and flux assumptions.

\bigskip 

\noindent \emph{Summary.}
Whether we consider only self-dual fluxes or all fluxes, the topology associated to~$(\mathcal{J},\chi)$ in examples from the Sch\"oller-Skarke database is rich enough to support the vacuum energy spacing~$< 10^{-120}$. Specifically, with all (self-dual) fluxes we find that that overwhelming majority of the Sch\"oller-Skarke database exhibits such small spacing ($99.96\%$ ($99.95\%$) when ${\cal V}^2\Lambda_0/\hat{\Pi}^2=1$), as measured by~$\delta$.

\section{Membrane Nucleation and Decay Rates}
\label{sec:cosmo}
Equipped with a simplified model for the cosmological constant, we may consider dynamical transitions between vacua via tunneling effects in cosmology.
Consider a flux changing transition via the Brown-Teitelboim mechanism, wherein 
\begin{equation}
N \mapsto N' = N + \Delta N\,.
\end{equation}
For simplicity, we assume the transition is mediated by a BPS or anti-BPS brane wrapped on a cycle determined by~$\Delta N$. Such membrane has a tension 
\begin{equation}
T = 4\pi \sqrt{\frac{2 g_s}{\cV}}\, \big\vert\hat \Pi\cdot \Delta N \big\vert\,,
\label{Tdef} 
\end{equation}
where our normalization convention for the holomorphic three-form~$\Omega$ is~$||\Omega|| = 1$. See Appendix~\ref{app:holomorphic_form} for a discussion of this point, and~\cite{deAlwis:2013gka} for a derivation using a different normalization. The associated change in vacuum energy is given by 
\begin{equation}
\Delta V := V_N - V_{N'} =  \frac{3}{\cV^2}  \big(\hat \Pi\cdot \Delta N\big) \, \big(2 \hat \Pi \cdot N  + \hat \Pi\cdot \Delta N \big)\,.
\label{Del V def}
\end{equation}
We focus on ``down-tunneling'' transitions, which lower the vacuum energy and therefore have~$\Delta V > 0$.\footnote{Up-tunneling are comparatively suppressed by an exponential of the difference in de Sitter entropies.} Hence both parenthetical factors must have the same sign. 

The tunneling rate is~$\Gamma \sim {\rm e}^{-B}$, where the bounce action for a Brown-Teitelboim transition in the thin-wall, probe-brane approximation, ignoring gravitational corrections, is given by
\begin{equation}
\label{eqn:BT_bounce}
B = \frac{27\pi^2}{2} \frac{T^4}{(\Delta V)^3} 
= 512 \pi^6  g_s^2 \,\mathcal{V}^4 \,\,\frac{\hat\Pi \cdot  \Delta N}{\big(2 \hat \Pi \cdot N + \hat \Pi\cdot \Delta N \big)^3} \,.
\end{equation}
We require~$B \gg 1$ to ensure self-consistency of the semi-classical approximation and validity of the dilute instanton gas approximation. 
Notice, however, that if~$\hat{\Pi}$ is commensurate, {\it i.e.},
if there exists~$\Delta N \in \bZ^J$ such that~$\hat \Pi  \cdot \Delta N = 0$, then~$B=0$ in this case,
violating our assumptions. 
Fortunately, commensurate vectors in~$\mathbb{R}^{\cal J}$ are a set of measure zero. Conversely, for~$\hat \Pi$ incommensurate,~$\hat \Pi \cdot \Delta N \neq 0$,  we have~$B > 0$
and the dilute instanton gas can always be made valid by an appropriate scaling of the volume~$\cV$. We henceforth restrict to this case in order to analyze bounce actions.

\medskip
Bubble nucleation via the decay of the false vacuum is dominated by the largest decay rates, corresponding to the smallest bounces. A key question is whether the dominant decay channel corresponds to a ``small step'' ({\it i.e.}, with~$\big\vert\Delta N\big\vert \sim 1$) or a ``giant leap'' ({\it i.e.},~$\big\vert \Delta N \big\vert\gg 1$) in flux space. In BP, the answer depends on the flux configuration of the initial state~\cite{Brown:2010mg}. In our case, {\it we will see that the bubble transitions are always dominated by giant leaps.} To see this, it is helpful to analyze the bounce action in two limiting regimes, according to which term in the denominator of Eq.~\eqref{eqn:BT_bounce} dominates.

\begin{itemize}

\item \textbf{Case One.} The first regime is~$\big\vert\hat \Pi \cdot N\big\vert \ll \big\vert\hat \Pi\cdot \Delta N \big\vert$, in which case
\begin{equation}
B \simeq 512 \pi^6  g_s^2 \,\mathcal{V}^4\frac{1}{\big(\hat \Pi\cdot \Delta N \big)^2} \,;\qquad \big\vert\hat \Pi \cdot N\big\vert \ll \big\vert\hat \Pi\cdot \Delta N \big\vert \,.
\label{B first regime} 
\end{equation}
Clearly the bounce action is minimized by maximizing~$\big\vert \hat \Pi\cdot \Delta N\big\vert$.
In turn, since~$\hat \Pi$ is positive,~$\big\vert \hat \Pi\cdot \Delta N\big\vert$ can be maximized by increasing~$| \Delta N|$, up to the D3 tadpole constraint that eventually cuts off the ability to increase it. To be precise, since~$\Delta N = N' - N$, then a fixed initial state~$N$ we see that~$|\Delta N|$ is maximized if~$N'$ is antipodal to~$N$, and saturates the inequality~\eqref{N bound}, {\it i.e.},~$|N'| = \sqrt{\chi/12}$. In other words,
\begin{equation}
\big\vert \Delta N\big\vert \leq R_{\Delta N} \equiv \sqrt{\frac{\chi}{12}} + |N|  \,.
\label{R Delta N}
\end{equation}
The bounce action~\eqref{B first regime} in this first regime thus satisfies
\begin{equation}
B \gtrsim 512 \pi^6  g_s^2 \,\mathcal{V}^4 \frac{1}{\hat{\Pi}^2 \left(\sqrt{\frac{\chi}{12}} + |N|\right)^2} \,. 
\end{equation}
Importantly, the bounce is minimized generically for large flux changes.

\item \textbf{Case Two.} The more subtle case is that of~$\big\vert\hat \Pi \cdot N\big\vert \gg \big\vert\hat \Pi\cdot \Delta N \big\vert$. In this regime we have 
\begin{equation}
B \simeq 64 \pi^6  g_s^2 \,\mathcal{V}^4 \frac{\hat \Pi \cdot \Delta N }{\big(\hat \Pi \cdot N\big)^3} \,;\qquad \big\vert\hat \Pi \cdot N\big\vert \gg \big\vert\hat \Pi\cdot \Delta N \big\vert\,.
\label{B second regime}
\end{equation}
Evidently, the bounce action is minimized by minimizing~$ \hat \Pi\cdot \Delta N$ in this case. For what type of flux changes does this occur? 
For a ``small step'' in flux space, defined to be~$\Delta N_{I,\pm} := \pm e_I$, where~$e_I$ is a unit vector in the~$I$-th direction, we have
\begin{equation}
\hat \Pi \cdot \Delta N_{I,\pm} = \pm \hat \Pi_I\,.
\end{equation}
For such small jumps the dependence of the bounce on~$\hat \Pi \cdot \Delta N$ is only through~$\hat \Pi_I$:
\begin{equation}
B_\text{small step} \simeq 64 \pi^6  g_s^2 \,\mathcal{V}^4  \frac{\hat \Pi}{\big(\hat \Pi \cdot N\big)^3} \frac{\hat{\Pi}_I}{\hat{\Pi}} \leq 64 \pi^6  g_s^2 \,\mathcal{V}^4  \frac{\hat \Pi}{\big(\hat \Pi \cdot N\big)^3}\,.
\label{B small step}
\end{equation}
 
On the other hand, by Dirichlet's approximation theorem (see Appendix~\ref{sec:Dirichlet}), we know that there exists a~$\Delta N_I$ such that 
\begin{equation}
\big\vert \hat \Pi\cdot \Delta N\big\vert \leq  \hat{\Pi} \frac{2^{{\cal J}-1} \Gamma(\frac{{\cal J}+1}{2})}{\pi^{\frac{{\cal J}-1}{2}}R_{\Delta N}^{{\cal J}-1}}\,,
\label{dirichlet result}
\end{equation}
where~$R_{\Delta N}$ is defined by~\eqref{R Delta N}. Such~$\Delta N$ generically describes a giant leap in flux space. Using Stirling's formula, the bounce action in this regime therefore satisfies
\begin{eqnarray}
B_\text{giant leap} &\leq&  64 \pi^7  g_s^2 \,\mathcal{V}^4  \frac{\hat \Pi}{\big(\hat \Pi \cdot N\big)^3} \sqrt{\frac{\chi}{24{\rm e}}}
\left(\frac{24{\cal J}}{{\rm e}\pi \chi}\right)^{{\cal J}/2} \left(\frac{1}{1 + \frac{|N|}{\sqrt{\chi/12}}}\right)^{{\cal J}-1} \,.
\label{B giant leap 2nd regime}
\end{eqnarray}
Equation~\eqref{B giant leap 2nd regime} differs from the inequality~\eqref{B small step} for small steps mainly by the last two factors.
The last factor is clearly less than unity, but more important is the factor of~$\left(\frac{24{\cal J}}{{\rm e}\pi \chi}\right)^{{\cal J}/2}$, which is~$\ll 1$ for~${\cal J} \gg 1$ and~${\cal J}/\chi \lesssim 1$. Since these conditions are necessary to obtain a small vacuum energy spacing, we conclude that giant leaps dominate vacuum transitions in this regime.

\end{itemize}

What remains to be shown is that these regimes can both be relevant. Taking~$\Delta N_I$ to be continuous and taking a first derivative, we have
\begin{equation}
\frac{\partial B}{\partial \Delta N_I} \propto \hat\Pi_I \,\, \frac{2\big(\hat\Pi\cdot N - \hat\Pi\cdot \Delta N\big)}{\big(2\hat\Pi \cdot N + \hat\Pi\cdot \Delta N \big)^4}\,,
\end{equation}
establishing that there is a critical point when~$\hat\Pi \cdot N =\hat\Pi\cdot \Delta N$. An analysis of the second derivative demonstrates that this critical point is a maximum. We therefore minimize~$B$ by moving away from~$\Pi \cdot N =\hat\Pi\cdot \Delta N$, yielding the two limiting regimes we have analyzed.\footnote{In moving away from~$\hat\Pi \cdot N =  \hat\Pi\cdot \Delta N$, one might violate the dilute instanton gas approximation before developing a large hierarchy between the quantities. This can be remedied by taking~$\mathcal{V}$ sufficiently large.}

We have argued that there are two limiting regimes relevant for making the bounce small, and that in both of the regimes the flux changes that minimize the bounce (and thereby maximize the decay rate) are large. Long jumps dominate tunneling in this model.

\subsection{Lifetime Considerations}

For a geometry to be phenomenologically viable, it should not only harbor vacua with sufficiently small vacuum energy, but these should also be sufficiently long lived.
Focusing on a single, dominant decay channel for simplicity, the tunneling rate per unit volume in the semi-classical approximation is
\be
\Gamma = M^4 {\rm e}^{-B}\,,
\ee
where~$B$ is the bounce action, and~$M^4$ is the fluctuation determinant. To be most conservative, we will take the latter to be set by the Planck scale,~$M \sim M_{\rm Pl}$.
Specifically, we are interested in estimating the lifetime of a low-lying dS vacuum, with~$V_\star \ll \Lambda_0$. Such a vacuum is metastable if the probability to nucleate a bubble within a Hubble time~$\sim H_\star^{-1}$ and over a Hubble volume~$\sim H_\star^{-3}$ is less than unity, where~$H_\star = \sqrt{\frac{V_\star}{3M_{\rm Pl}^2}}$. That is,
\be
\frac{\Gamma}{H_\star^4}\lesssim 1  \qquad \Longrightarrow \qquad \left(\frac{3M_{\rm Pl}^4}{V_\star}\right)^2 {\rm e}^{-B} \lesssim 1\,.
\label{lifetime bound 1}
\ee

From~\eqref{hyperplane eqn}, low-lying dS vacua have
\be
\big\vert \hat{\Pi} \cdot N\big\vert \simeq \sqrt{\frac{{\cal V}^2\Lambda_0}{3}}\,.
\label{Pi dot N low lying}
\ee
Since~$\big\vert \hat{\Pi} \cdot N\big\vert$ is large, the dominant transitions out of such vacua are in all likelihood in the second regime considered above, namely $\big\vert\hat \Pi \cdot N\big\vert \gg \big\vert\hat \Pi\cdot \Delta N \big\vert$. The relevant bounce action inequality is given by Eq.~\eqref{B giant leap 2nd regime}. By Cauchy-Schwarz inequality,~\eqref{Pi dot N low lying} implies~$\hat{\Pi} |N| \geq  \sqrt{\frac{{\cal V}^2\Lambda_0}{3}}$, and therefore
\be
 \frac{|N|}{\sqrt{\chi/12}} \geq  \sqrt{\frac{4{\cal V}^2\Lambda_0}{\hat{\Pi}^2 \chi}} \,.
\label{N over chi bound}
 \ee
The inequality ~\eqref{B giant leap 2nd regime} is about the existence of a decay channel with bounce bounded above. Plugging
~\eqref{Pi dot N low lying} and~\eqref{N over chi bound} into it, we obtain
\be
B\lesssim \frac{192 \pi^7}{\sqrt{2{\rm e}}}  \frac{g_s^2 {\cal V}^2}{\Lambda_0} \left(1 + \sqrt{\frac{\hat{\Pi}^2\chi}{4{\cal V}^2\Lambda_0}} \right) \left(\frac{\sqrt{\frac{24{\cal J}}{{\rm e}\pi \chi}}}{1 + \sqrt{\frac{4{\cal V}^2\Lambda_0}{\hat{\Pi}^2\chi}}} \right)^{\cal J}\,,
\label{B bound interm}
\ee
This bound applies to low-lying dS vacua.
Combining this result with the lifetime bound~\eqref{lifetime bound 1}, which bounds the bounce from below, we obtain
\be
\boxed{\frac{g_s^2 {\cal V}^2}{\Lambda_0} \left(1 + \sqrt{\frac{\hat{\Pi}^2\chi}{4{\cal V}^2\Lambda_0}} \right) \left(\frac{\sqrt{\frac{24{\cal J}}{{\rm e}\pi \chi}}}{1 + \sqrt{\frac{4{\cal V}^2\Lambda_0}{\hat{\Pi}^2\chi}}} \right)^{\cal J} \;\gtrsim\; \frac{\sqrt{2{\rm e}}}{96\pi^7} \ln \frac{3M_{\rm Pl}^4}{V_\star} \; \simeq \;  10^{-3} \,,}
\label{lifetime bound 2}
\ee
where in the last step we have substituted~$V_\star \simeq 10^{-120} M_{\rm Pl}^4$ for our vacuum. 

The primary danger of not having a long enough lifetime arises from the term with the $\mathcal{J}$ exponent, leading to
\begin{equation}
    \frac{24\mathcal{J}}{e\pi \chi} > 1
    \label{eqn:approx lifetime}
\end{equation}
as an effective necessary condition in the case that $\mathcal{J}>1$ and the $\sqrt{4\mathcal{V}^2\Lambda_0/(\hat\Pi^2 \chi)}= \sqrt{4\tilde \Lambda_0/(\hat\Pi^2 \chi \mathcal{V}^\alpha)}$ in the denominator is negligible assuming large enough $\chi$ and $\mathcal{V}$. For instance, if we assume~$\frac{{\cal V}^2\Lambda_0}{\hat{\Pi}^2} = 1$, as before, then~$\frac{{\cal V}^2\Lambda_0}{\hat{\Pi}^2\chi} \ll 1$ for~$\chi \gg 1$. Further assuming $g_s=1$ and applying this to the Sch\"oller-Skarke database, we find that the entire set of $532,600,401$ fourfolds with $\chi>0$ (all are $\geq 0$, only $82$ have $\chi=0$)  satisfy both the effective condition \eqref{eqn:approx lifetime} and  \eqref{lifetime bound 2}, justifying the approximate condition. Furthermore, the $\chi=0$ cases are pathological, since it change the nature of the D3 tadpole condition.

\section{Conclusions}
\label{sec:conclusions}

In this paper we have studied a modified BP model motivated by Type IIB and F-theory flux compactifications, with vacuum energy given by 
\begin{equation}
V = \Lambda_0 - \frac{3}{\mathcal{V}^2}\, |N_I \Pi_I|^2\,,
\end{equation}
where we allow $\Lambda_0$ to have volume dependence as $\Lambda_0 = \tilde \Lambda_0 / \mathcal{V}^\alpha$.
The key structural difference relative to the original BP model is the sign of the flux contribution: rather than raising the cosmological constant from a negative bare value, fluxes lower it from a positive uplift term~$\Lambda_0$. This simple sign change has significant consequences for both the geometry of the discretuum and the dynamics of vacuum transitions.
Our model is related to work by Denef and Douglas~\cite{Denef:2004ze,Denef:2004cf}, but its simplicity allows us to carry out analyses of vacuum energy spacing and bubble nucleation in the spirit of BP.

We have shown that the locus of small dS cosmological constants in flux space takes the form of thin wafers rather than thin shells. The wafer geometry arises because equipotentials are hyperplanes orthogonal to the period vector~$\hat{\Pi}$, intersected with the D3-brane tadpole ball. This geometric picture enabled an analysis of vacuum energy spacing using a saddlepoint approximation for counting lattice points, which we demonstrated is essential when the flux dimension~$\mathcal{J}$ is large relative to the tadpole charge~$\chi$. Applying this framework to the Sch\"oller-Skarke database of~$532,600,483$ distinct sets of Hodge numbers for Calabi-Yau fourfolds, we find that the overwhelming majority of cases exhibit vacuum energy spacings of~$10^{-120}$ or smaller (e.g., $99.95\%$ for specific parameter values), demonstrating that the topological richness of these compactifications is more than sufficient to accommodate the observed cosmological constant. 

We further analyzed Brown-Teitelboim membrane nucleation in this model, finding that the bounce action admits a maximum along the hyperplane~$\hat{\Pi} \cdot N = \hat{\Pi} \cdot \Delta N$. Moving asymptotically away from this hyperplane in either direction minimizes the bounce and maximizes the decay rate. In both asymptotic regimes, the dominant vacuum decays correspond to large jumps in flux space, a conclusion supported by Dirichlet's approximation theorem applied to the lattice of flux quanta. This preference for large flux changes has potential implications for the cosmological measure problem, as it suggests that the dynamics of the landscape may favor transitions that traverse significant distances in the discretuum. 
We obtain a bound 
\begin{equation}
    \frac{24 \mathcal{J}}{e\pi\chi} > 1
\end{equation}
on the topology of the Calabi-Yau necessary for a low-lying dS vacuum to have a lifetime consistent with the age of the universe. It is satisfied for the Sch\"oller-Skarke database.

\medskip
Several avenues for future work present themselves. First, the moduli dependence we have suppressed deserves systematic treatment. In a complete analysis, the periods~$\Pi_I$ are evaluated at critical points in complex structure moduli space, and the wafer geometry would be replaced by a more intricate structure reflecting the interplay between flux choice and moduli stabilization. Second, our analysis of membrane nucleation utilized the thin-wall, probe-brane, no-gravity approximation; incorporating gravitational corrections and thick-wall effects would sharpen predictions for decay rates and vacuum lifetimes. Additionally, one must consider the dependence of the bounce action on the moduli space metric, potential, and vacuum locations. Third, the preference for large flux jumps we have identified should be incorporated into cosmological measure proposals. It would be interesting to understand whether this dynamical feature, combined with anthropic selection effects, leads to predictions for the distribution of observed cosmological constants. Fourth, the role of vertical fluxes merits investigation, particularly in connection with gauge symmetry breaking and the Standard Model. Finally, it would be interesting to study the computational complexity of our model in light of proposed complexity measures \cite{Carifio:2017nyb,Denef:2017cxt,Khoury:2019yoo,Khoury:2019ajl,Kartvelishvili:2020thd,Khoury:2022ish,Khoury:2023ktz,Hassfeld:2022vzk,Hassfeld:2024ake} in cosmology.

We hope that this simplified model, in the spirit of the original BP proposal, will serve as a useful conceptual tool for exploring the interplay between the string landscape, cosmological dynamics, and the cosmological constant problem.

\bigskip
\noindent\textbf{Acknowledgements.} We thank Mehmet Demirtas, Jakob Moritz, Liam McAllister, Wati Taylor, and Yi-Nan Wang for discussions. We are grateful to Friedrich Sch\"oller for providing Hodge data from the Sch\"oller-Skarke database. J.H. is supported by NSF CAREER grant PHY-1848089 and PHY-2209903. The work of J.K. is supported in part by the DOE (HEP) Award DE-SC0013528.  The work of C.L. supported in part by the Alfred P. Sloan Foundation Grant No. G-2019-12504 and by DOE Grant DE-SC0013607.

\let\cleardoublepage\clearpage
\appendix

\section{Dirichlet's Approximation Theorem}
\label{sec:Dirichlet}

Dirichlet's theorem is a fundamental result in Diophantine approximation. The original theorem is for a single variable, but we wish to use the multi-dimensional generalization.

The multi-dimensional version can  can be obtained as a corollary of Minkowski's theorem, which states that every convex set in~$\bR^n$ that is symmetric about the origin and has  volume greater than~$2^n$ contains a non-zero integer point~$N \in \mathbb{Z}^n$. Let us use this result to prove Dirichlet's approximation theorem:
\begin{theorem}[Dirichlet's Approximation Theorem]
For any~$\Pi\in\mathbb{R}^n$, there exists a constant~$C$ such that for any~$Q>0$, there exists a nonzero~$ N\in\mathbb{Z}^n$ with~$\| N\|\le Q$ and
\[
  |\Pi\cdot N| \;\leq\; \frac{C}{Q^{\,n-1}}.
\]
\end{theorem}
\begin{proof}
Consider the convex set~$K=\{x\in \bR^n \, : \, ||x||\leq Q,\, \,\, |\Pi\cdot x | \leq \epsilon\}$ where~$\epsilon > 0$. The volume of this set is 
\begin{equation}
  \text{vol}(K) = \frac{2\epsilon}{||\Pi||}\,\,  \text{vol}(B_{n-1}(Q)) = \frac{2\epsilon}{||\Pi||}\, \frac{\pi^{\frac{n-1}{2}}}{\Gamma(\frac{n+1}{2})} \, Q^{n-1},
\end{equation}
and therefore~$K$ satisfies the assumptions of Minkowski's theorem if
\begin{equation}
\epsilon > ||\Pi|| \frac{2^n \Gamma(\frac{n+1}{2})}{2 \pi^{\frac{n-1}{2}}} \frac{1}{Q^{n-1}}=: \frac{C}{Q^{n-1}}.
\label{epsilon C}
\end{equation}
This establishes the existence of a non-zero~$N\in \bZ^n$ with~$|\Pi \cdot N| \leq \frac{C}{Q^{n-1}} + \delta =: \epsilon_\delta$ and~$\|N\|\leq Q$, for all~$\delta > 0$. Taking the limit~$\delta \to 0$ gives the desired result, technically requiring an infinitesimal~$\delta$ on the RHS. 
\end{proof}

\section{Lattice Points in High Dimensional Spheres}
\label{app:volume_approximation}

\begin{figure}[thbp]
   \begin{subfigure}{0.45\textwidth}
       \centering
       \includegraphics[width=\linewidth]{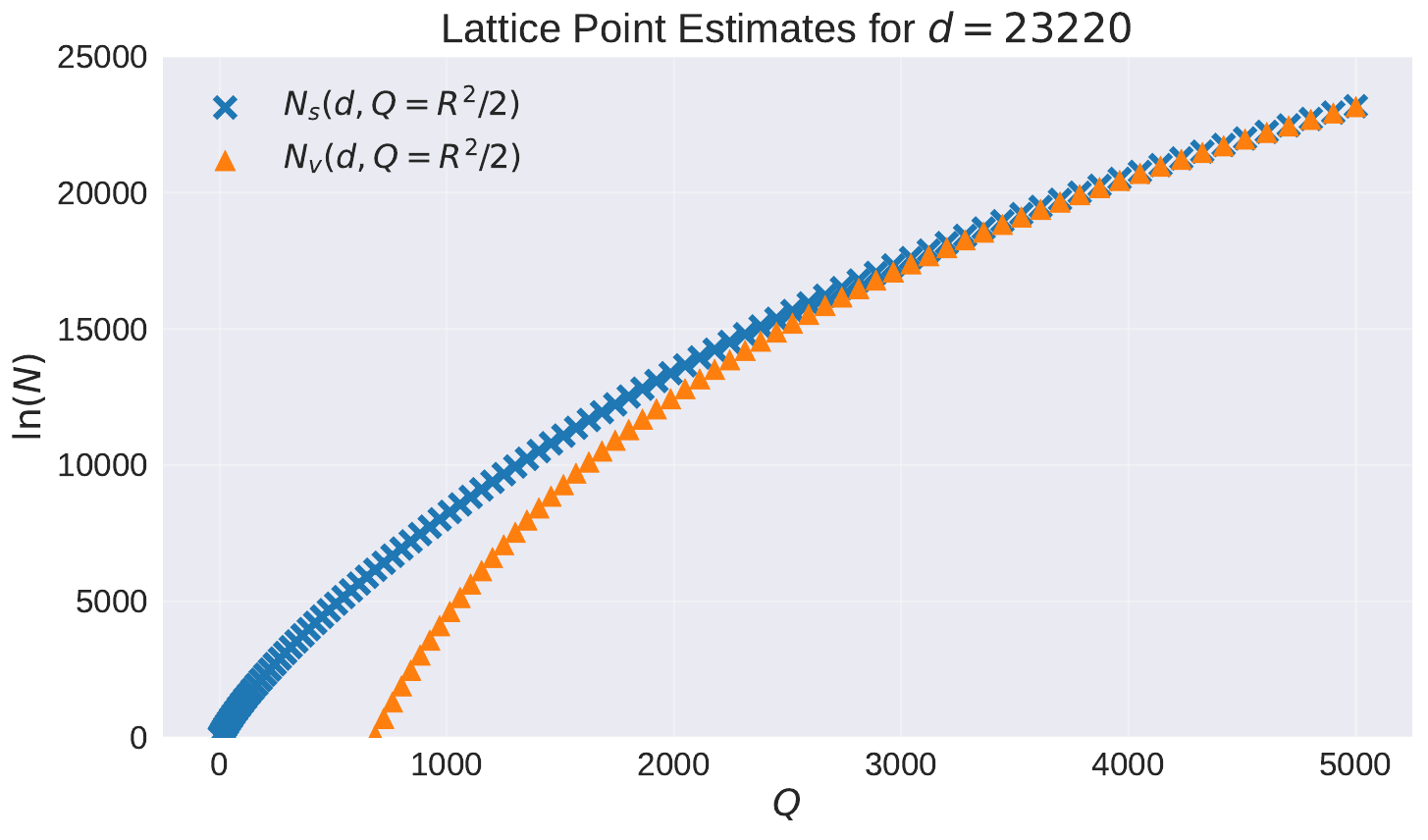}
   \end{subfigure}
   \hfill
   \begin{subfigure}{0.45\textwidth}
       \centering
       \includegraphics[width=\linewidth]{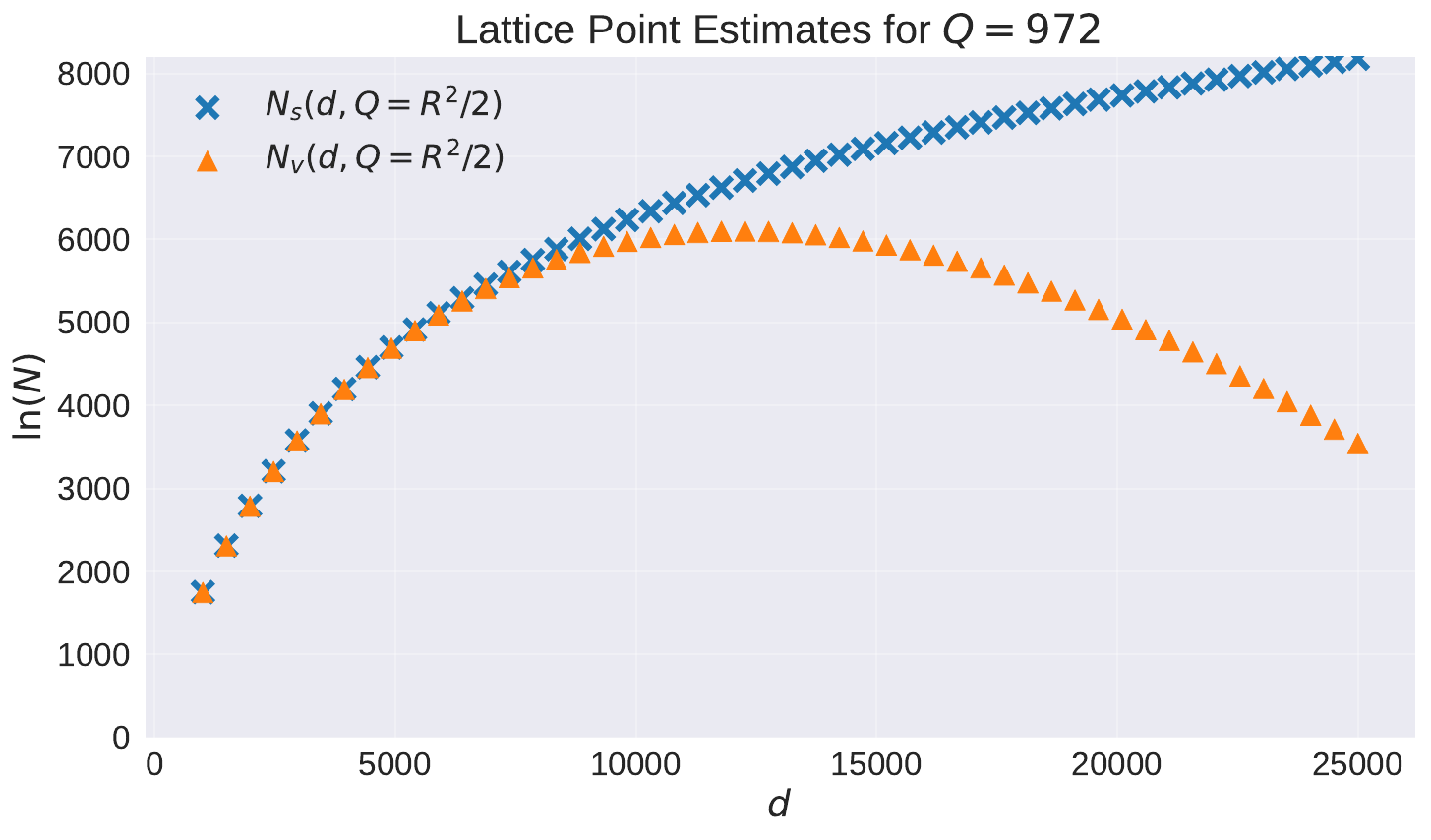}
   \end{subfigure}
   \caption{Lattice point estimates comparing saddlepoint and volume methods. Left: Fixed dimension~$d = 23220$ with varying~$Q$. Right: Fixed~$Q = 972$ with varying dimension~$d$.}
   \label{fig:denef_comparison}
\end{figure}

In this Section we collect facts and perform some elementary computations about the number of lattiec points in~$B_{d,R}$, the~$d$-dimensional ball of radius~$R$. 

Let~$N(d,R)$ be the number of lattice points in~$B_{d,R}$, and let~$N_{\rm saddle}(d,R)$ and~$N_{\rm vol}(d,R)$ denote respectively the saddlepoint and volume approximations to~$N(d,R)$ introduced in Sec.~\ref{sec:genBP}. To motivate the importance of these analyses we would like to reproduce the results of Fig.~15 in \cite{Denef:2008wq}. Taking~$d=23,220$ and a variety radii~$R$ and associated~$Q=R^2/2$, we find the results in Fig.~\ref{fig:denef_comparison}. We see that for fixed~$d$ and sufficiently large~$Q$, the volume approximation agrees with the saddlepoint approximation. On the other hand, for fixed~$Q$ and varying~$d$, we see that the volume approximation eventually \emph{decreases} with increasing~$d$, in clear contradiction to the expected growth of the number of lattice points. For~$d$ that is large relative to~$R$, the volume approximation breaks down, as highlighted in Sec.~\ref{sec:genBP} in different estimates for the F-theory geometries with the most flux vacua and the largest~$h^{1,1}$. From the LHS we see that the volume approximation is good provided that~$Q\gtrsim d/8$, as noted in \cite{Denef:2008wq}. One may verify this estimate also for~$d=b_4^{\text{max}} = 1819942$ and~$b_4^{h^{1,1}\text{-max}} = 1214150$, but unfortunately both of these examples have~$Q \ll d/8$, and therefore the volume approximation is not valid. One must use the saddlepoint method for these values of~$d, R$.

The saddlepoint method for counting lattice points in higher dimension spheres is introduced in~\cite{Mazo1990}, where the use of the saddlepoint method appears near the end, and error estimates for this approximation are not given. From the saddle point action~\eqref{eqn:saddle_point_action} we see that the saddlepoint approximation is valid provided that~$d$ is large enough. However, we wish to get some idea of how large~$d$ must be. To do so, one may obtain exact counts~$N(d,R)$ for~$d\leq 10$ and~$R\leq 18$ from the Online Encyclopedia of Integer Sequences (OEIS)~\cite{OEIS}. We can then compare the saddlepoint approximation, the volume approximation, and the exact counts for small~$d$ and~$R$. In doing so, we see that the absolute percent error in the base~$10$ exponent relative to the brute force result approaches about~$15\%$ and~$2\%$ for the saddlepoint and volume approximations, respectively, as~$d$ increases to~$10$. This should be contrasted with, {\it e.g.}, the~$\sim 200\%$ discrepancy in the exponent between the volume and saddle approximations for the F-theory geometry with the most flux vacua, \eqref{eqn:max_comparison}. For these reasons, and since we will generally be interested in large~$d$ for the sake of robust flux landscapes, we utilize the saddlepoint approximation.

\bibliographystyle{chetref}
\section{Normalization of Holomorphic Top-form}
\label{app:holomorphic_form}

Let $\Omega$ be the holomorphic three-form on a Calabi-Yau three-fold $X$, $[\Sigma]$ a class in the middle homology of $X$, and $\Sigma$ a calibrated submanifold (a special Lagrangian). $\Omega$ is unique up to scale, but the physics in this paper requires being precise about the normalization. Specifically, the volume of a general submanifold $\Sigma'$ is given by 
\begin{equation}
\text{Vol}(\Sigma') = \int_{\Sigma'} \sqrt{g} \, {\rm d}^n x\,.
\end{equation}
On the other hand, if $\Sigma$ is calibrated by $\Omega$, then
\begin{equation}
\text{Vol}(\Sigma) \propto \text{Re} \int_\Sigma \Omega \,.
\end{equation}
The normalization of $\Omega$ is crucial to determine the proportionality constant between the volume and the period $\int \Omega$, which in this paper appear in the physics of the Brown-Teitelboim membrane tension and the flux superpotential, respectively. Two common choices for normalization are 
\begin{enumerate}
\item $\int_X \Omega \wedge \overline{\Omega} = 1$, common in the mathematics literature. This normalization is less convenient for physics, since it does not have a natural scaling with the volume of $X$.
\item $\int_X \Omega \wedge \overline{\Omega} = -8{\rm i}\, \text{Vol}(X)$, in which case $\text{Vol}(\Sigma) = \text{Re} \int_\Sigma \Omega$; see \cite{Koerber:2010bx}.
\end{enumerate}
Putting the pieces together, we have that 
\begin{equation}
\text{Vol}(\Sigma) = \frac{\sqrt{8 \mathcal{V}}}{||\Omega||} \, \text{Re} \int_\Sigma \hat \Omega
\end{equation}
in a general normalization, where $\mathcal{V}$ is the volume of $X$.

However, since the K\"ahler potential for the complex structure sector is $-\log \left({\rm i} \int_X \Omega \wedge \overline{\Omega}\right)$, which would yield a $1/||\Omega||^2$ dependence in the potential, we see that our simplified model for the vacuum energy \eqref{our V} takes $||\Omega|| = 1$, and therefore $\text{Vol}(\Sigma) = \sqrt{8 \mathcal{V}} \, \text{Re} \int_\Sigma \hat \Omega$. This normalization choice is a gauge choice and does not affect the physics of the model.

\begingroup
\begin{center}
\textbf{References.}
\end{center}
\renewcommand{\section}[2]{}%
\bibliography{refs}
\endgroup

\end{document}